\documentclass[12pt]{amsart}
 \pdfoutput=1
\usepackage{amsmath}
\usepackage{latexsym}
\usepackage{amsfonts}
\usepackage{amssymb}
\usepackage{color}
\usepackage{bbm,dsfont}
\usepackage{graphicx}
\usepackage{hyperref}
\usepackage{enumerate}
\usepackage{subfigure}
\usepackage{MnSymbol}
\usepackage{enumitem}

\newtheorem{proposition}{Proposition}
\newtheorem{theorem}{Theorem}

\newtheorem{corollary}{Corollary}

\theoremstyle{definition}

\newtheorem{remark}{Remark}
\newtheorem{example}{Example}
\newtheorem{definition}{Definition}

\newcommand{\real}{\mathbb R} 
\newcommand{\RR}{\mathbb R} 
\newcommand{\complex}{\mathbb C} 
\newcommand{\CC}{\mathbb C} 
\newcommand{\half}{\tfrac{1}{2}} 
\newcommand{\mo}[1]{\left| #1 \right|} 
\newcommand{\abs}{\mo} 

\newcommand{\ff}{\mathcal{F}}

\newcommand{\hi}{\mathcal{H}} 
\newcommand{\hh}{\mathcal{H}} 
\newcommand{\kk}{\mathcal{K}} 
\newcommand{\lh}{\mathcal{L(H)}} 
\newcommand{\lk}{\mathcal{L(K)}} 
\newcommand{\elle}[1]{\mathcal{L}(#1)} 
\newcommand{\ip}[2]{\left\langle\,#1\,|\,#2\,\right\rangle} 
\newcommand{\no}[1]{\left\|#1\right\|} 
\newcommand{\tr}[1]{\mathrm{tr}\left[#1\right]} 
\newcommand{\id}{\mathbbm{1}} 

\newcommand{\ind}[1]{\mathfrak{i}(#1)}

\newcommand{\A}{\mathsf{A}}
\newcommand{\B}{\mathsf{B}}
\newcommand{\C}{\mathsf{C}}
\newcommand{\D}{\mathsf{D}}
\newcommand{\F}{\mathsf{F}}
\newcommand{\G}{\mathsf{G}}
\newcommand{\X}{\mathsf{X}}
\newcommand{\Y}{\mathsf{Y}}
\newcommand{\Z}{\mathsf{Z}}

\newcommand{\e}{{\rm e}}
\newcommand{\tA}{A}
\newcommand{\tB}{B}
\newcommand{\Sym}{{\rm Sym}}
\newcommand{\Bin}[2]{\left(\begin{array}{c} #1 \\ #2 \end{array}\right)}

\begin{document}

\title{Quantum incompatibility in collective measurements}

\begin{abstract}{We study the compatibility (or joint measurability) of quantum observables in a setting where the experimenter has access to multiple copies of a given quantum system, rather than performing the experiments on each individual copy separately. We introduce the index of incompatibility as a quantifier of incompatibility in this multi-copy setting, as well as the notion of compatibility stack representing the various compatibility relations present in a given set of observables. We then prove a general structure theorem for multi-copy joint observables, and use it to prove that all abstract compatibility stacks with three  vertices have realizations in terms of quantum observables. }
\end{abstract}

\author[Carmeli]{Claudio Carmeli}
\address{\textbf{Claudio Carmeli}; DIME, Universit\`a di Genova, Via Magliotto 2, I-17100 Savona, Italy}
 \email{claudio.carmeli@gmail.com}

\author[Heinosaari]{Teiko Heinosaari}
\address{\textbf{Teiko Heinosaari}; Turku Centre for Quantum Physics, Department of Physics and Astronomy, University of Turku, Finland}
\email{teiko.heinosaari@utu.fi}

\author[Reitzner]{Daniel Reitzner}
\address{\textbf{Daniel Reitzner}; Institute of Physics, Slovak Academy of Sciences, D\'ubravsk\'a cesta 9, 845 11 Bratislava, Slovakia}
\email{daniel.reitzner@savba.sk}

\author[Schultz]{Jussi Schultz}
\address{\textbf{Jussi Schultz}; Turku Centre for Quantum Physics, Department of Physics and Astronomy, University of Turku, Finland}
\email{jussi.schultz@gmail.com}

\author[Toigo]{Alessandro Toigo}
\address{\textbf{Alessandro Toigo}; Dipartimento di Matematica, Politecnico di Milano, Piazza Leonardo da Vinci 32, I-20133 Milano, Italy, and I.N.F.N., Sezione di Milano, Via Celoria 16, I-20133 Milano, Italy}
\email{alessandro.toigo@polimi.it}

\maketitle




\section{Introduction}

The laws of quantum physics dictate that there are certain tasks which are mutually exclusive, meaning that they cannot be performed simultaneously with a single device. 
This quantum incompatibility is usually encountered in the context of mutually exclusive measurements: one cannot measure two orthogonal spin directions $\sigma_x$ and $\sigma_y$ with a single measurement setup. In the modern quantum information point of view, incompatibility has been identified as a genuine resource gained from switching from classical to quantum protocols. As such, it is at the heart of many typically quantum applications, such as secure quantum key distributions, or the possibility to steer remote quantum systems. Due to its importance in these applications as well as its status as a fundamental quantum feature, it is essential to gain a deeper understanding of quantum incompatibility.

Even though quantum theory gives predictions of outcomes in statistical experiments, the essence of incompatibility is manifested on the level of single experimental runs. More specifically, for compatible observables it is possible to find a single measurement setup such that on each experimental run, the reading of the measurement outcome allows one to assign the values of the outcomes for the compatible observables. The prototypical example of this is the joint measurement of a pair of compatible observables $\A$ and $\B$: if $\Omega_\A$ and $\Omega_\B$ are the outcome sets for $\A$ and $\B$, respectively, then a joint observable will have the outcome set $\Omega_\A\times \Omega_\B$. The measurement outcome in each experimental run is therefore a pair of numbers $(a,b)$, from which we assign the values $a$ and $b$ as the outcomes  of $\A$ and $\B$. By repeating this procedure multiple times, the resulting distributions should then correspond to those obtained from the separate statistical experiments of $\A$ and $\B$. This should highlight the distinction between the joint measurement of a pair of observables, and any scenario where the outcome distributions are reconstructed from the full distribution of a third measurement.

In this paper we take a step away from this usual framework, and study joint measurements of observables in a setting where the experimenter has access to multiple copies of a given quantum system, rather than performing the experiments on each individual copy separately. At first sight it may seem that the whole phenomenon of incompatibility is lost in such an approach: if two copies of the same system are available, then by measuring $\sigma_x$ on one system and $\sigma_y$ on the other, one has in a sense measured these incompatible observables jointly. However, things change drastically when one looks at more than two observables. In fact, by including also a third spin direction, $\sigma_z$, one gets a triple of incompatible observables which cannot be jointly measured even with two copies of the same system.

This approach leads to a new way of treating and quantifying the incompatibility of larger sets of observables, by looking at the minimal number of system copies needed to be able to measure all of them with a single collective measurement. We will define this number to be the {\em index of incompatibility} of the set. On a more detailed level, we define the {\em compatibility stack} of a set of observables as a list of hypergraphs expressing the various multi-copy compatibility relations between the observables. This definition naturally generalizes the joint measurability hypergraphs introduced in \cite{KuHeFr14}. After these general treatments we focus on the qubit case, where we prove results for the multi-copy joint measurability of triples of noisy qubit observables. In particular, we demonstrate that all compatibility stacks of order 3 have a quantum realization.

The paper is organized as follows. In Section \ref{sec:observables} we recall the definition of a quantum observable. Section \ref{sec:k-compatibility}  presents the notion of $k$-compatibility. The definition is then expanded in Section \ref{sec:stack}  to define the compatibility stack --- a mathematical way of describing $k$-compatibility relations within a given set of observables. Section \ref{sec:structure}  goes deeper into the notion of $k$-compatibility and provides a necessary and sufficient condition for $k$-compatibility of $n$ observables. The general content of the previous parts is then exemplified in Section \ref{sec:qubitcase} in the case of three qubit observables.
The conclusion and future outlooks are given in Section~\ref{sec:conclusion}.

\section{Quantum observables}
\label{sec:observables}

We start by recalling the definition of a quantum observable as a positive operator valued measure. 
In this paper we will restrict our investigation to observables with finite number of measurement outcomes. 
We refer to \cite{OQP97} for an exhaustive presentation of the properties of quantum observables.

The quantum mechanical description of a physical system is based on a complex Hilbert space $\hi$, which we assume to be finite dimensional throughout the paper. 
We denote by $\lh$ the vector space of linear operators on $\hi$, in which we let $\id$ be the identity operator.

\begin{definition}\label{def:obs}
Let $\Omega$ be a finite set. A map $\A:\Omega \to \lh$ is an \emph{$\lh$-valued observable on $\Omega$} if
\begin{enumerate}[leftmargin=*,labelsep=3mm,label=(\roman*)]
\item $\A(x) \geq 0$ for all $x\in\Omega$;
\item $\sum_{x\in\Omega} \A(x) = \id$.
\end{enumerate}
\end{definition}
The states of the system are represented by positive trace one operators on $\hh$, and for a  state $\varrho$ the number $\tr{\varrho\A(x)}$ is the probability of obtaining an outcome $x$ in a measurement of $\A$. 

As an example, consider the $x$-component of the spin of a spin-$1/2$ system. The corresponding observable is then the map $\mathsf{X}:\{-1,+1\}\to\mathcal{L}(\complex^2)$ defined on the two-outcome set $\Omega=\{+1,-1\}$ and having as its values the two orthogonal projections

$$
\mathsf{X}(\pm 1) = \frac{1}{2} (\id \pm \sigma_x).
$$
We can  add white noise to this observable, and this results in a noisy spin observable $\mathsf{X}_a:\{+1,-1\}\to \mathcal{L}(\complex^2)$ defined as

$$
\mathsf{X}_a(\pm 1) = \frac{1}{2} (\id \pm a\sigma_x),
$$
where the parameter $1-a$ is the noise intensity.
The $y$- and $z$-components of the spin are of course treated in the same manner giving rise to the corresponding observables $\mathsf{Y}$ and $\mathsf{Z}$, and their noisy versions \(\mathsf{Y}_b\) and \(\mathsf{Z}_c\).

It is occasionally convenient to view an observable $\A$ as a map on the power set $2^\Omega$ rather than the set $\Omega$.  For any $X\subseteq\Omega$, we denote $\A(X)=\sum_{x\in X} \A(x)$ so that 
\begin{enumerate}[leftmargin=*,labelsep=3mm,label=(\roman*')]
\item $\A(X) \geq 0$ for all $X\subseteq\Omega$;
\item $\A(\Omega) = \id$;
\item $\A(X\cup Y) = \A(X) + \A(Y)$ for all $X,Y\subseteq\Omega$ such that $X\cap Y = \emptyset$.
\end{enumerate}
The two definitions are clearly equivalent, and we will switch between them whenever it is convenient.

\section{$k$-compatibility of observables}\label{sec:k-comp}
\label{sec:k-compatibility}

\subsection{Definition}

Let $\A_1,\ldots,\A_n$ be $\lh$-valued observables with outcome sets $\Omega_1,\ldots,\Omega_n$.
The compatibility of these observables means that we can simultaneously implement their measurements, even if only one input state is available.
Generalizing the usual formulation of joint measurements, we assume that we have access to $k$ copies of the initial state. 
We can hence make a collective measurement on a  state $\varrho^{\otimes k}$ (For any $A\in\lh$ we use the notation $A^{\otimes k}\in\mathcal{L}(\hi^{\otimes k})$ for the $k$-fold tensor product $A^{\otimes k} = A\otimes \ldots \otimes A$, and we set $\hh^{\otimes 0}=\complex$  and $A^{\otimes 0} = 1$.). This measurement should give a measurement outcome for each observable $\A_1,\ldots,\A_n$, so we are looking for an $\elle{\hi^{\otimes k}}$-valued observable $\G$ on the product set $\Omega_1\times\ldots\times\Omega_n$.
In order for $\G$ to serve as a joint measurement, it is required that if we ignore other than the $i$th component $x_i$ of a measurement outcome $(x_1,\ldots,x_n)$, the probability must agree with the probability of getting $x_i$ in a measurement of $\A_i$.

For this reason, we introduce the $i$th marginal $\G^{[i]}$ of $\G$.
For all $X\subseteq\Omega_i$, $\G^{[i]}$ is the observable given by

\begin{align*}
\G^{[i]}(X) &= \G(\Omega_1 \times\cdots\times\Omega_{i-1}\times X \times\Omega_{i+1}\times\cdots\times\Omega_n) \, \\
& = \G(\pi_i^{-1}(X)) ,\nonumber
\end{align*}
where $\pi_i:\Omega_1\times\ldots\times\Omega_n\to\Omega_i$ is the projection $\pi_i(x_1,\ldots, x_n)= x_i$. 
This definition of a marginal can also be written in an equivalent form as 
\begin{equation*}
\G^{[i]}(x) = \mathop{\sum_{x_1,\ldots,x_{i-1},}}_{x_{i+1},\ldots,x_n} \G(x_1,\ldots,x_{i-1},x,x_{i+1},\ldots,x_n) \, 
\end{equation*}
required to hold for all $x\in\Omega_i$.

\begin{definition}\label{prop:marginal}
$\lh$-valued observables $\A_1,\ldots,\A_n$ on the outcome sets $\Omega_1,\ldots,\Omega_n$, respectively, are \emph{$k$-compatible} if there exists an $\elle{\hh^{\otimes k}}$-valued observable $\G$ on the product set $\Omega_1\times\ldots\times\Omega_n$, such that
\begin{equation}\label{eq:marginal}
\tr{\varrho^{\otimes k} \G^{[i]}(x)} = \tr{\varrho \A_i(x)}
\end{equation}
for all $i=1,\ldots,n$, $x \in \Omega_i$ and all states $\varrho$.
The observable $\G$ is called a \emph{$k$-copy joint observable} of $\A_1,\ldots,\A_n$.
\end{definition}

If $k=1$ in Definition \ref{prop:marginal}, then we have the usual definition of compatibility, also called joint measurability \cite{Lahti03}.

\subsection{Basic properties}

Let us observe some basic properties of the $k$-compatibility relation.
Firstly, for observables $\A_1,\ldots,\A_n$ we can define
\begin{equation}
\label{eq:kcompJM}
\G(x_1,\ldots,x_n) = \A_1(x_1) \otimes \ldots \otimes \A_n(x_n)
\end{equation}
and this $\elle{\hh^{\otimes n}}$-valued observable clearly satisfies \eqref{eq:marginal} with $k=n$.
We thus conclude that
\begin{itemize}[leftmargin=*,labelsep=4mm]
\item \emph{Any collection of $n$ observables is $n$-compatible.}
\end{itemize}
Secondly, if $\G$ is a $k$-copy joint observable of $\A_1,\ldots,\A_n$, we get a $k$-copy joint observable of $\A_1,\ldots,\A_{n-1}$ by simply summing over the outcomes in $\Omega_n$.
More generally, we have that
\begin{itemize}[leftmargin=*,labelsep=4mm]
\item \emph{Any subset of a $k$-compatible set of observables is $k$-compatible}.
\end{itemize}
 Finally, if $\G$ is a $k$-copy joint observable, we can trivially extend it to a higher dimensional Hilbert space by setting $\G' = \G \otimes \id$.
Therefore, we conclude that
\begin{itemize}[leftmargin=*,labelsep=4mm]
\item \emph{Any collection of $k$-compatible observables is $k'$-compatible for all $k'\geq k$.}
\end{itemize}

The fourth simple but important property of the $k$-compatibility relation is the following additivity.

\begin{proposition}\label{prop:division}
Let $\mathcal{A}$ be a finite collection of observables and $\mathcal{A}=\mathcal{A}_1 \cup \mathcal{A}_2$ for two nonempty subsets $\mathcal{A}_1$ and $\mathcal{A}_2$.
If $\mathcal{A}_1$ is $k_1$-compatible and $\mathcal{A}_2$ is $k_2$-compatible, then $\mathcal{A}$ is $(k_1+k_2)$-compatible.
\end{proposition}

\begin{proof}
First, if $\mathcal{A}_i=\mathcal{A}$ for some $i$, then the claim is trivial.
Hence, we assume that $\mathcal{A}_i \neq \mathcal{A}$ for all $i\in\{1,2\}$.
We denote $\mathcal{A}_3=\mathcal{A}_2\smallsetminus \mathcal{A}_1$.
As $\mathcal{A}_3$ is a subset of $\mathcal{A}_2$, it is $k_2$-compatible.
The set $\mathcal{A}$ is a disjoint union of $\mathcal{A}_1$ and $\mathcal{A}_3$, and we can label the observables such that $\mathcal{A}_1=\{\A_1,\ldots,\A_m\}$ and $\mathcal{A}_3=\{\A_{m+1},\ldots,\A_n\}$.
We denote by $\G_1$ and $\G_3$ the $k_1$- and $k_2$- copy joint observables of $\mathcal{A}_1$ and $\mathcal{A}_3$, respectively, and then define

\begin{equation*}
\G(x_1,\ldots,x_n) = \G_1(x_1,\ldots,x_m) \otimes \G_3(x_{m+1},\ldots,x_n) \, .
\end{equation*}
This observable is a $(k_1+k_2)$-copy joint observable of $\mathcal{A}$.
\end{proof}

\subsection{Index of incompatibility}\label{sec:index}

For any set of $n$ observables, it is well defined  the smallest integer $1\leq k \leq n$ such that the collection is $k$-compatible.
This leads to the following notion.

\begin{definition}
\label{def:indexofinc}
The \emph{index of incompatibility} is the minimal number of copies that is needed in order to make a given set of observables compatible.
Hence, for a set of observables $\mathcal{A}$ the index of incompatibility $\ind{\mathcal{A}}$ is

\begin{equation*}
\ind{\mathcal{A}} := \min_k \{ \textrm{$\mathcal{A}$ is $k$-compatible} \} \, .
\end{equation*}
\end{definition}

The usual compatibility corresponds to 1-compatibility, hence the index of incompatibility of a compatible set of observables is $1$.
The index of incompatibility can be taken as an integer valued quantification of the incompatibility of a given set. 
Our earlier observations and Proposition \ref{prop:division} imply the following.
\begin{enumerate}[leftmargin=*,labelsep=3mm,label=(\roman*)]
\item $1 \leq \ind{\mathcal{A}} \leq \#\mathcal{A}$;
\item if $\mathcal{A}\subseteq\mathcal{B}$, then  $\ind{\mathcal{A}} \leq \ind{\mathcal{B}}$;
\item $\ind{\mathcal{A}_1 \cup \mathcal{A}_2} \leq \ind{\mathcal{A}_1} + \ind{\mathcal{A}_2}$;
\item $\ind{\mathcal{A}}=1$ if and only if $\mathcal{A}$ is compatible.
\end{enumerate}

It is not clear from this definition if for each integer $n\geq 2$ there exists a set $\mathcal{A}$ of $n$ observables such that the index $\ind{\mathcal{A}}$ has the maximal value $n$. 
In Section \ref{sec:qubitcase} we will show that there exists a triplet of observables whose index of incompatibility is $3$.

\section{Compatibility stack}
\label{sec:stack}

\subsection{Definition}
Although the index of incompatibility gives a simple quantification of the incompatibility of a set of observables, it does not take into account the finer compatibility structures present in the set. This calls for a more refined description of the various compatibility relations between the observables. In the usual single copy scenario, this can be conveniently done in terms of \emph{joint measurability hypergraphs} \cite{KuHeFr14}. 

In general, a \emph{hypergraph} is a pair $(V,E)$ consisting of a set $V$ and a set $E$ of non-empty subsets of $V$.
The elements of $V$ are called \emph{vertices} and the elements of $E$ \emph{edges}. 
Following \cite{KuHeFr14}, we say that a hypergraph $(V,E)$ is a joint measurability hypergraph (or compatibility hypergraph) if all non-empty subsets of edges are also edges, i.e.,
\begin{equation*}
\emptyset \neq \mathcal{A}' \subseteq \mathcal{A} \in E \Rightarrow \mathcal{A}' \in E \, .
\end{equation*}
Every set of observables gives rise to a joint measurability hypergraph where the vertices represent the observables, and the edges linking some particular vertices represent the compatibility of the corresponding observables.  The above condition then states that compatibility of some set of observables implies compatibility of any subset of these observables. 
Furthermore, it was shown in \cite{KuHeFr14} that every abstract joint measurability hypergraph where all the singleton sets are edges has such a realization in terms of quantum observables.

The generalization of this approach to the case of $k$-compatibility is given by the following notion.

\begin{definition}
\label{def:stack}
Let $V$ be a finite set with $n$ elements and let $E_k$ be a set of non-empty subsets of $V$ for $k=1,\ldots,n$.
We denote $H_k=(V,E_k)$.
The list $(H_1,\ldots,H_n)$ of hypergraphs is a \emph{compatibility stack} if
\begin{itemize}[leftmargin=*,labelsep=4mm]
\item[(S1)] each $H_k=(V,E_k)$ is a joint measurability hypergraph,
\item[(S2)] $E_1$ contains all singleton sets and $E_n =2^V$, and
\item[(S3)] if $\mathcal{A}\in E_k$ and $\mathcal{B}\in E_l$, then $\mathcal{A}\cup\mathcal{B}\in E_{k+l}$.
\end{itemize}
\end{definition}

The motivation for the previous definition is that 
\emph{any finite set of observables gives rise to a compatibility stack}.
Namely, let $V=\{\A_1, \ldots, \A_n\}$ be a finite set of observables.
We take these observables as vertices, and a set of edges $E_k$ is defined in a way that a subset $\mathcal{A}\subseteq V$ belongs to $E_k$ if $\mathcal{A}$ is $k$-compatible.
The conditions (S1)--(S3) hold by our earlier discussion. 
First, every subset of a $k$-compatible set of observables is also $k$-compatible, hence $(V,E_k)$ is a joint measurability hypergraph.
Second, each set made of one observable is $1$-compatible, hence $E_1$ contains all singleton sets. 
The condition $E_n=2^V$ follows from the facts that any set of $n$ observables is $n$-compatible, and any subset of $n$-compatible observables is also $n$-compatible.
Finally, Proposition \ref{prop:division} is reflected in condition (S3). 

\begin{proposition}
\label{prop:stack}
For a compatibility stack $(H_1,\ldots,H_n)$ the following hold:
\begin{itemize}[leftmargin=*,labelsep=3mm]
\item[(1)] $E_1 \subseteq E_2 \subseteq \cdots \subseteq E_n$
\item[(2)] For each $k=1,\ldots,n$, the set $E_k$ contains all subsets of $V$ of order $k$.
\end{itemize}
\end{proposition}

\begin{proof}
(1) Fix $k=2,\ldots,n$. Suppose $\mathcal{A}\in E_{k-1}$, and pick $\A\in\mathcal{A}$. We have $\{\A\}\in E_1$ by (S2).
Then $\mathcal{A} = \mathcal{A} \cup \{\A\}$ and hence $\mathcal{A}\in E_{k}$ by (S3).
Therefore, $E_{k-1} \subseteq E_{k}$.

(2) This follows by  induction. Indeed, by (S2) the claim is true for $E_1$. If $\mathcal{A}$ is an order $k+1$ subset of $V$ and $\A\in\mathcal{A}$, then $\mathcal{A}\setminus\{\A\}\in E_k$ by the inductive hypothesis, hence $\mathcal{A} = \{\A\}\cup(\mathcal{A}\setminus\{\A\}) \in E_{k+1}$ by (S3).
\end{proof}

The condition (1) abstractly reflects the understanding that if a set of observables is $k$-compatible, it is $(k+1)$-compatible as well. The condition (2) is, on the other hand, saying that any collection of $k$ observables is $k$-compatible.

For a compatibility stack $(H_1,\ldots,H_n)$ consisting of hypergraphs $H_k=(V,E_k)$, we say that the \emph{index} of a non-empty subset $\mathcal{A}\subseteq V$ is the smallest integer $j$ such that $\mathcal{A}\in E_j$.
If the compatibility stack represents the $k$-compatibility relations of a set of observables, then the index of $\mathcal{A}$ is exactly the index of incompatibility as given by Definition~\ref{def:indexofinc}. It is clear that the normalization, monotonicity and subadditivity properties (i)--(iii) of the index of incompatibility are still retained in this abstract setting.

\subsection{Compatibility stacks with three vertices}

\begin{figure}
    \centering
    \subfigure[]
    {
         \includegraphics[width=3cm]{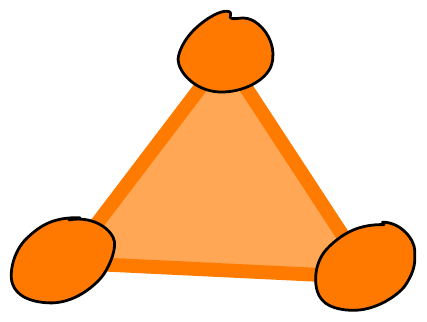}      }
    \subfigure[]
    {
        \includegraphics[width=3cm]{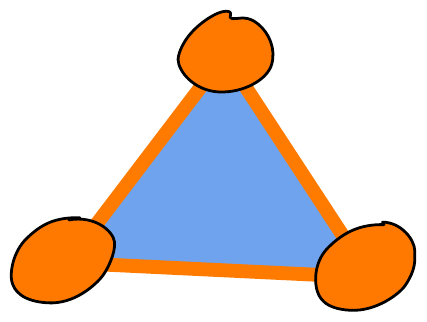} 
    }
     \subfigure[]
    {
        \includegraphics[width=3cm]{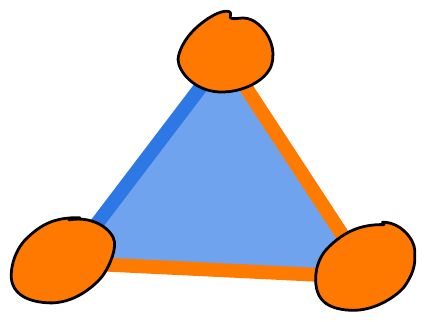} 
    } \subfigure[]
    {
        \includegraphics[width=3cm]{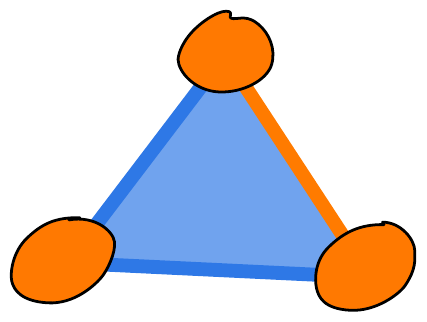} 
    } \subfigure[]
    {
        \includegraphics[width=3cm]{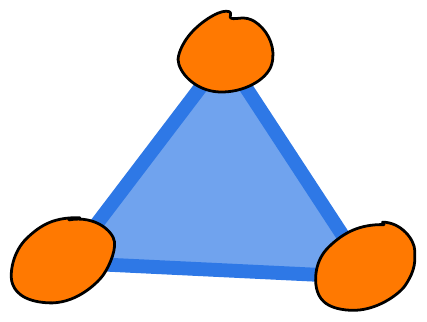} 
    }
     \subfigure[]
    {
        \includegraphics[width=3cm]{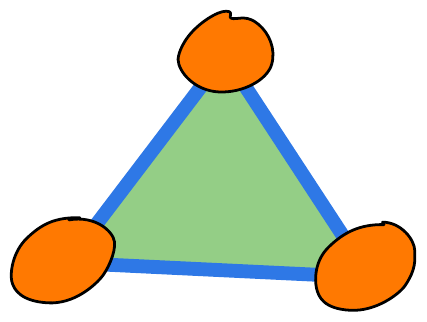} 
    }
        \caption{All possible compatibility stacks with three vertices. Orange color marks index 1, blue marks index 2 and green marks index 3.}
    \label{fig:possible}
\end{figure}

The simplest (non-trivial) example of a compatibility stack is the case of three vertices $\A$, $\B$, $\C$.
A graphical representation of a compatibility stack $(H_1,H_2,H_3)$ is a triangle, where the edges and the area can be colored according to the corresponding index. 
The situation is depicted in Figure~\ref{fig:possible} for all possible compatibility stacks and in Figure~\ref{fig:impossible} for some impossible cases.

For example, in the case of Figure~\ref{fig:possible}.(a), the compatibility stack is given by $H_1^{(\mathrm{a})}=H_2^{(\mathrm{a})}=H_3^{(\mathrm{a})}=(V,2^V)$ with $V=\{\A,\B,\C\}$ being the set of vertices.
For the cases (b) through (e) we still have $H_2^{\text{(b)--(e)}}=H_3^{\text{(b)--(e)}}=(V,2^V)$ but for $H_1^{\text{(b)--(e)}}$ its set $E_1$ has fewer and fewer elements. In case (e) $H_1^{(\mathrm{e})}$ becomes simply $H_1^{(\mathrm{e})}=(V,\{\{\A\},\{\B\},\{\C\}\})$. Finally, in the case (f) the compatibility stack is given as
\begin{align*}
H_1^{(\mathrm{f})} &= (V,\{\{\A\},\{\B\},\{\C\}\}),\\
H_2^{(\mathrm{f})} &= (V,\{\{\A\},\{\B\},\{\C\},\{\A,\B\},\{\A,\C\},\{\B,\C\}\}),\\
H_3^{(\mathrm{f})} &= (V,2^V).
\end{align*}

The case of Figure~\ref{fig:impossible}.(a) is on first sight representable by a stack
\begin{align*}
H_1^{(\mathrm{a})\prime} &= (V,\{\{\A\},\{\B\},\{\C\},\{\A,\B\}\}),\\
H_2^{(\mathrm{a})\prime} &= (V,\{\{\A\},\{\B\},\{\C\},\{\A,\B\},\{\A,\C\},\{\B,\C\}\}),\\
H_3^{(\mathrm{a})\prime} &= (V,2^V).
\end{align*}
However, its impossibility comes from the fact that $E_1$ contains both $\{\C\}$ and $\{\A,\B\}$, which by (S3) would require $E_2$ to contain the set $\{\A,\B,\C\}$, which is not the case. Similar discussion is valid also for the case (b). 

The  fact that some collections of hypergraphs are not  compatibility stacks comes from the fact that Definition \ref{def:stack} puts limitations on the indeces of the hypergraph edges. In particular, condition (S3) reduces the number of possible compatibility stacks more than Proposition \ref{prop:stack} alone. 
We shall make this clear in the discussion of four vertices below.

\begin{figure}
    \centering
    \subfigure[]
    {
         \includegraphics[width=3cm]{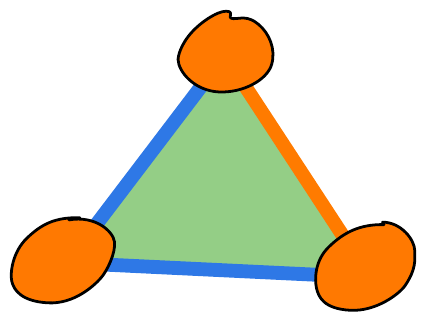}      }
    \subfigure[]
    {
        \includegraphics[width=3cm]{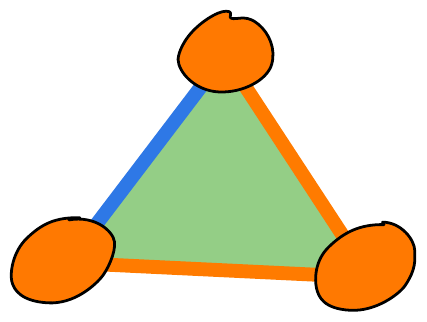} 
    }
    \caption{Examples of impossible $k$-compatibility relations for three observable. Orange color marks index 1, blue marks index 2 and green marks index 3.}
    \label{fig:impossible}
\end{figure}

\subsection{Compatibility stacks with four vertices}

Having four vertices increases the number of possible compatibility stacks considerably. 
Let the four vertices be denoted as $\A$, $\B$, $\C$ and $\D$. 
For any pair of these vertices, the remaining pair will be called \emph{reciprocal}, e.g.~$\{\B,\C\}$ is reciprocal to $\{\A,\D\}$.
The four vertices can be illustrated as the vertices of a tetrahedron. In this representation, different types of graph edges\footnote{In this section, edges of the graph will be always denoted as \emph{graph edges}, while the physical edges of the tetrahedron will be just \emph{edges}.} correspond to different elements of the tetrahedron (vertices, edges, sides and bulk), with each of these possibly having a different index. We will abuse the language a bit by saying that particular elements of the tetrahedron are $k$-compatible, meaning that the corresponding graph edges have index $\leq k$.

The coarsest classification of compatibility stacks is by the index of the bulk.
The case of index-1 bulk is possible only when also all sides and edges have index 1 (Figure~\ref{fig:4possible}.(a)), since from Proposition \ref{prop:stack}.(1) we have $2^V=E_1\subseteq E_2 \subseteq \cdots \subseteq E_n$ implying the equality of all the sets $E_k$.
On the other hand, the bulk can have index 4 only when all sides have index 3 and edges have index 2 (Figure~\ref{fig:4possible}.(b)). Indeed, if on the contrary some edge would have index 1, then, by $2$-compatibility of its reciprocal edge, the bulk would be $3$-compatible by (S3). A similar consideration holds for the side indexes.

\begin{figure}
    \centering
    \subfigure[]
    {
         \includegraphics[width=4cm]{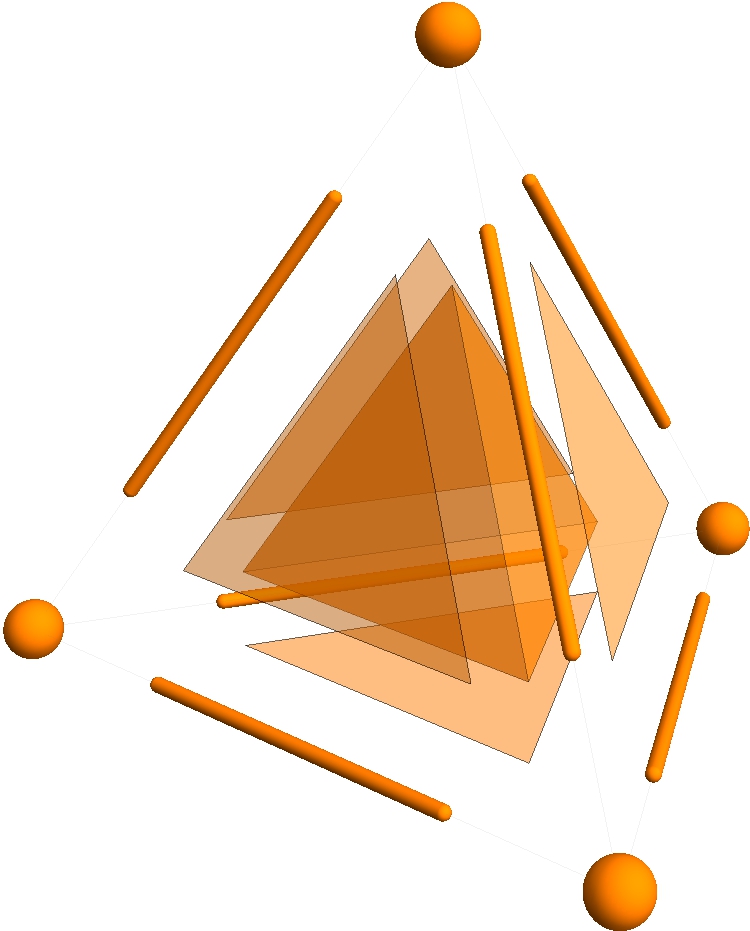}      }
    \subfigure[]
    {
         \includegraphics[width=4cm]{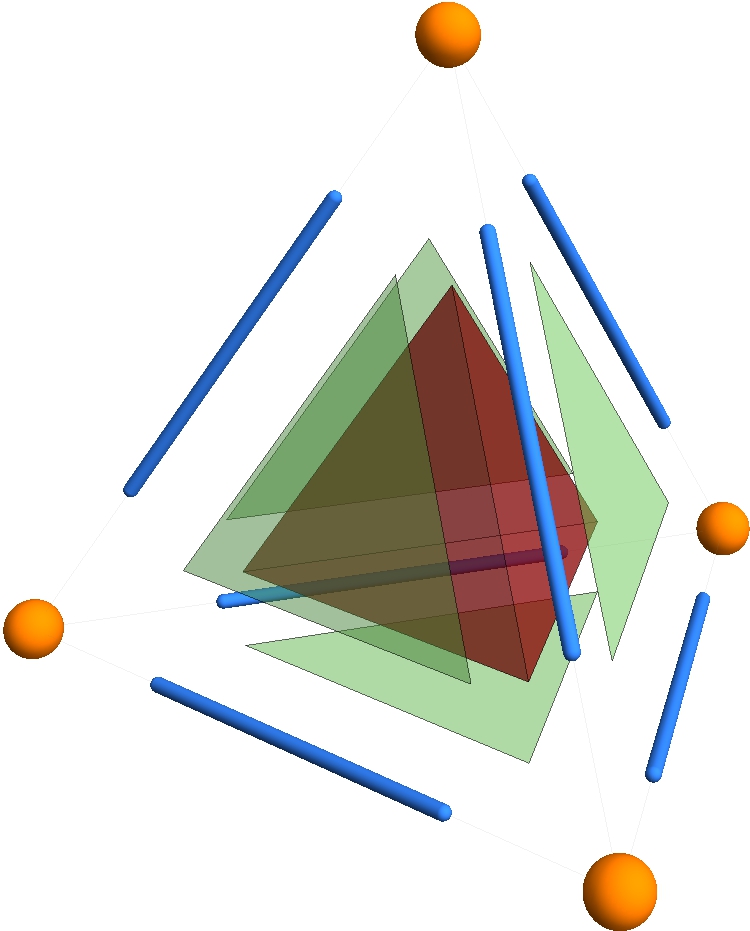}      }
    \subfigure[]
    {
         \includegraphics[width=4cm]{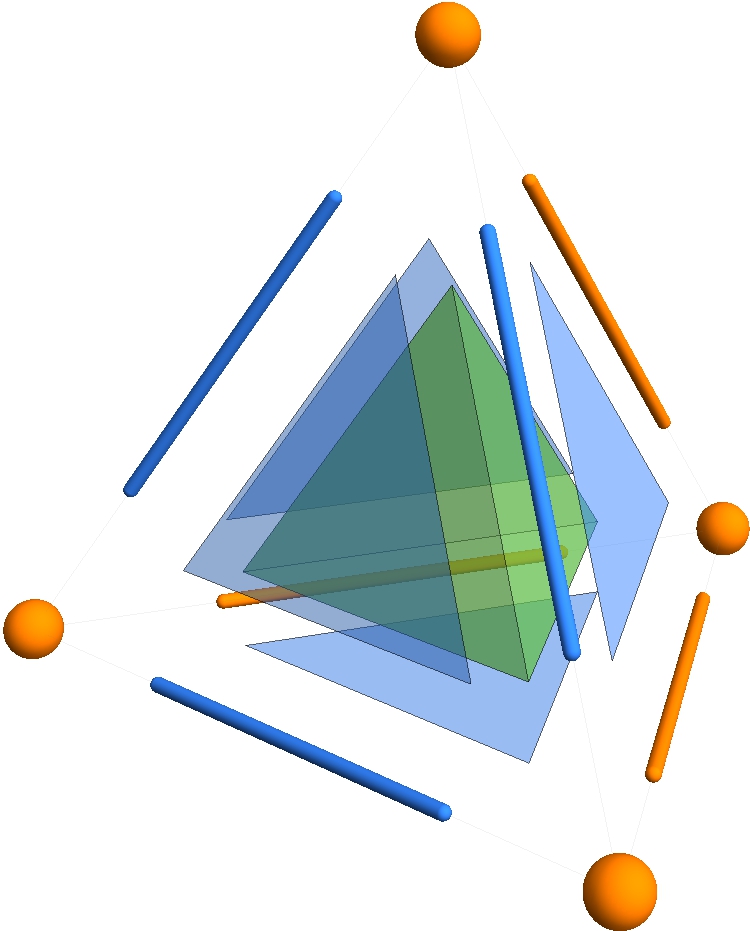}      }
    \subfigure[]
    {
        \includegraphics[width=4cm]{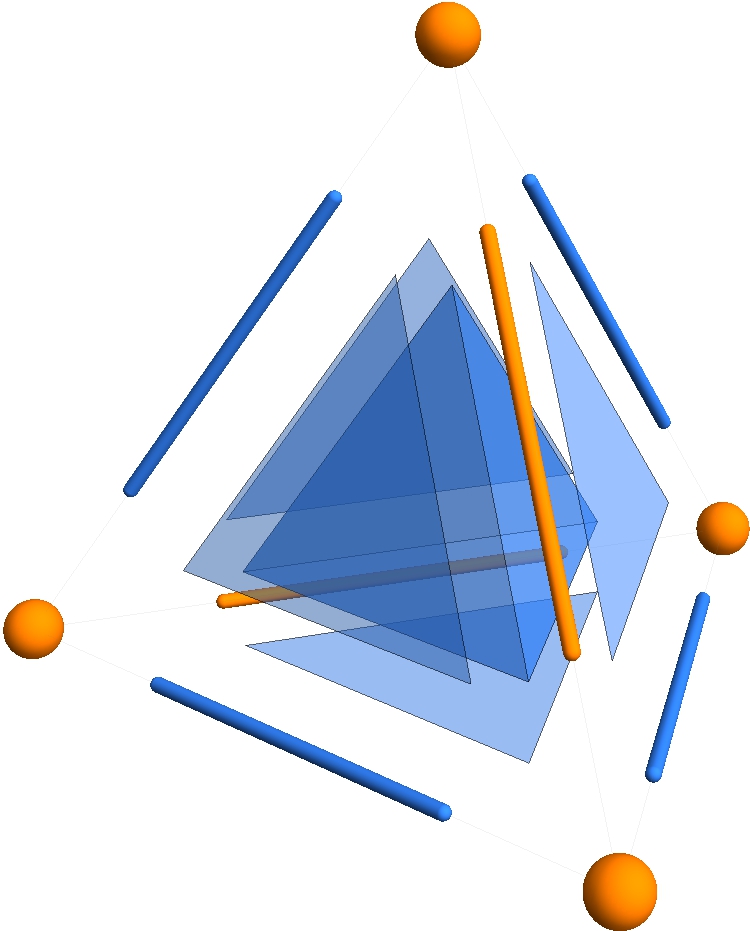} 
    }
    \caption{Compatibility stacks with four vertices can be represented by colored tetrahedrons. 
    As before, orange color marks index 1, blue marks index 2 and green marks index 3. In addition, index 4 is marked by red color.}
    \label{fig:4possible}
\end{figure}

The cases in between (as e.g.~Figure~\ref{fig:4possible}.(c)) are more populated and their number is reduced by the compatibility stack condition (S3). 
Particularly useful is the following consequence of the subadditivity of the index.
\begin{itemize}[leftmargin=*,labelsep=4mm]
\item \emph{If two sets composed of  reciprocal pairs have index 1, then the set of all four vertices has index $\leq 2$.}
\end{itemize}
Visually on the tetrahedron this means that if two opposing edges are compatible, then the bulk is 2-compatible (see e.g.~Figure~\ref{fig:4possible}.(d)). This observation is also easily intuitively grasped, as, if we have two pairs of compatible observables -- let us say $\{\B,\C\}$ and $\{\A,\D\}$ --, then there exist corresponding joint observables $\F$ and $\G$ respectively. These two observables are always $2$-compatible and, hence, also the four observables $\A$, $\B$, $\C$ and $\D$ are 2-compatible.

Definition~\ref{def:stack} of compatibility stack and Proposition~\ref{prop:stack} lead to the possibilities depicted in Table \ref{tab:4obs}. 
However, it is an open question if all these compatibility stacks have realizations by quantum observables.

\begin{table}
\begin{center}
\begin{tabular}{|c|ccccccc|}
\hline
\# of index-2 edges $\filledmedtriangleright$ & \ 0\ & \ 1\ & \ 2\ & \ 3\ & \ 4\ & \ 5\ & \ 6\ \\
Bulk index $\filledmedtriangledown$ & & & & & & & \\
\hline
1 & 1 & -- & -- & -- & -- & -- & --\\
2 & 5 & 3 & 3 & 4 & 2 & 1 & 1\\
3 & -- & -- & -- & 3 & 2 & 3 & 5\\
4 & -- & -- & -- & -- & -- & -- & 1\\
\hline
\end{tabular}
\end{center}

\medskip

\caption{All possible compatibility stacks enumerated by their bulk index and number of edges with index 2.  Altogether 34 different stacks (up to trivial permutations) are possible.}
\label{tab:4obs}
\end{table}

\section{Structure of $k$-copy joint observables}
\label{sec:structure}

In this section, we show that, in order to find the index of incompatibility of any collection of observables $\mathcal{A} = \{\A_1,\ldots,\A_n\}$, it is enough to characterize all {\em symmetric} $k$-copy joint observables of $\A_1,\ldots,\A_n$. In particular, we prove that the $k$-compatibility of $\A_1,\ldots,\A_n$ is equivalent to the usual compatibility of their symmetrized versions $\widetilde{\A}_1,\ldots,\widetilde{\A}_n$ in $\elle{\hh^{\otimes k}}$. This reduces the $k$-compatibility problem to a standard compatibility problem on an enlarged quantum system.

\subsection{Symmetric product}

The symmetric group $S_k$ acts in a natural way on the tensor product $\hh^{\otimes k}$ of $k$ copies of $\hh$: if $p\in S_k$ is any permutation, its action on a decomposable element $\psi_1\otimes\ldots\otimes\psi_k\in\hh^{\otimes k}$ is defined as
$$
\sigma(p)(\psi_1\otimes\ldots\otimes\psi_k) = \psi_{p^{-1}(1)}\otimes\ldots\otimes\psi_{p^{-1}(k)} \,.
$$
The map $\sigma:S_k\to\elle{\hh^{\otimes k}}$ is a unitary representation of $S_k$ on $\hh^{\otimes k}$.
 
Using this unitary representation, we then define the {\em symmetrizer channel} $\Sigma_k$ on $\elle{\hh^{\otimes k}}$ as

\begin{equation*}
\Sigma_k(\tA) = \frac{1}{k!} \sum_{p\in S_k}\sigma(p)\tA\sigma(p)^* \qquad \tA\in\elle{\hh^{\otimes k}} \, .
\end{equation*}
This map is completely positive and unital, hence it is a quantum channel. On decomposable operators $\tA = A_1\otimes\ldots\otimes A_k$, we have
$$
\Sigma_k(A_1\otimes\ldots\otimes A_k) = \frac{1}{k!} \sum_{p\in S_k} A_{p(1)}\otimes\ldots\otimes A_{p(k)} \,,
$$
hence $\Sigma_k$ is an idempotent projection onto the linear subspace $\Sym(k,\lh)$ of the $k$-symmetric tensor operators in $\elle{\hh^{\otimes k}} = \lh^{\otimes k}$. It becomes an orthogonal projection by endowing $\elle{\hh^{\otimes k}}$ with the Hilbert-Schmidt inner product $\ip{\tA}{\tB}_{HS} = \tr{\tA^*\tB}$.

The {\em symmetric product} of two operators $\tA_1\in\Sym(k_1,\lh)$ and $\tA_2\in\Sym(k_2,\lh)$ is the operator $\tA_1\odot\tA_2\in\Sym(k_1+k_2,\lh)$ with
$$
\tA_1\odot\tA_2 = \Sigma_{k_1+k_2}(\tA_1\otimes\tA_2)\,.
$$
The symmetric product is associative and commutative. 

We will constantly use the following, easily verifiable, formula:
if $A_1,\ldots,A_k,B_1,\ldots,B_k\in\lh$, then
$$
\ip{A_1\odot\ldots\odot A_k}{B_1\odot\ldots\odot B_k}_{HS} = \frac{1}{k!} \sum_{p\in S_k} \tr{A_1 B_{p(1)}}\cdots\tr{A_k B_{p(k)}} \,.
$$

\subsection{Structure theorem}
We will now prove the following theorem which shows that for a $k$-compatible set of observables, the set of $k$-copy joint observables always contains a symmetric observable.

\begin{theorem}\label{prop:k-comp}
The $\lh$-valued observables $\A_1,\ldots,\A_n$ on $\Omega_1,\ldots,\Omega_n$ are $k$-compatible if and only if there exists a $\elle{\hh^{\otimes k}}$-valued observable $\widetilde{\G}$ on $\Omega_1\times\ldots\times\Omega_n$ such that
\begin{equation}\label{eq:joint-2}
\widetilde{\G}^{[i]}(x) = \id^{\otimes (k-1)}\odot\A_i(x) \qquad \forall i=1,\ldots,n,\, x\in\Omega_i \, .
\end{equation}
In this case, we can choose $\widetilde{\G}$ such that $\widetilde{\G}(x_1,\ldots,x_n)\in\Sym(k,\lh)$ for all $x_1,\ldots,x_n$.
\end{theorem}

\begin{proof}
Sufficiency is easy, because any observable $\widetilde{\G}$ as in \eqref{eq:joint-2} satisfies
\begin{align*}
\tr{\varrho^{\otimes k}\widetilde{\G}^{[i]}(x)} & = \ip{\varrho^{\otimes k}}{\Sigma_k(\id^{\otimes (k-1)}\otimes\A_i(x))}_{HS} \\
& = \ip{\Sigma_k(\varrho^{\otimes k})}{\id^{\otimes (k-1)}\otimes\A_i(x)}_{HS} \\
& = \ip{\varrho^{\otimes k}}{\id^{\otimes (k-1)}\otimes\A_i(x)}_{HS} \\
& = \tr{\varrho\A_i(x)} \,,
\end{align*}
which is \eqref{eq:marginal}.

Conversely, suppose that \eqref{eq:marginal} holds for $\G$, and let $\widetilde{\G} = \Sigma_k\circ\G$. Then $\widetilde{\G}$ is an observable on $\Omega_1\times\ldots\times\Omega_n$ which is such that $\widetilde{\G}(x_1,\ldots,x_n)\in\Sym(k,\lh)$ for all $x_1,\ldots,x_n$. Denote $G = \widetilde{\G}^{[i]}(x)$. We have
\begin{equation*}
\begin{aligned}
\tr{\varrho^{\otimes k} G} & = \ip{\varrho^{\otimes k}}{\Sigma_k(\G^{[i]}(x))}_{HS} = \ip{\varrho^{\otimes k}}{\G^{[i]}(x)}_{HS} \\
& = \tr{\varrho\A_i(x_i)}
\end{aligned}
\end{equation*}
by \eqref{eq:marginal}. Choosing the state $\varrho = \id/d + t\Delta$, where $\abs{t}\leq  1/(d\no{\Delta})$ and $\Delta = \Delta^*$ with $\tr{\Delta} = 0$, the last equation gives
\begin{align*}
& \sum_{j=0}^k \Bin{k}{j} \tr{((\id/d)^{\otimes (k-j)}\odot\Delta^{\otimes j}) G} \, t^j = \\
& \qquad \qquad \qquad \qquad = \tr{(\id/d)\A_i(x)} + \tr{\Delta\A_i(x)}\, t \,.
\end{align*}
Comparing the coefficients of the same degree in $t$, we obtain the system of equations
\begin{align}
& d^{-k}\tr{G} = d^{-1}\tr{\A_i(x)} \label{eq:syst1}\\
& d^{-(k-1)}k\,\ip{\id^{\otimes (k-1)}\odot\Delta}{G}_{HS} = \ip{\Delta}{\A_i(x)}_{HS} \label{eq:syst2}\\
& \Bin{k}{j} d^{-(k-j)} \ip{\id^{\otimes (k-j)}\odot\Delta^{\otimes j}}{G}_{HS} = 0 \qquad \forall j\in\{2,3,\ldots,k\} \label{eq:syst3}
\end{align}
which must hold for all $\Delta\in\elle{\hh^{\otimes k}}$ with $\Delta^* = \Delta$ and $\tr{\Delta} = 0$. Now, take a set of $D=d^2-1$ opeartors $T_1,\ldots,T_D\in\lh$ such that $T_r^* = T_r$, $\tr{T_r} = 0$ and $\tr{T_r T_s} = \delta_{rs}$ for all $r,s=1,\ldots,D$, and write $\Delta = x_1 T_1+\ldots+x_D T_D$ for $x_1,\ldots,x_D\in\real$. Then \eqref{eq:syst3} yields
$$
\sum_{j_1+\ldots+j_D = j} \Bin{j}{j_1 \ \ldots \ j_D} x_1^{j_1}\ldots x_D^{j_D}\ip{\id^{\otimes (k-j)}\odot T_1^{\otimes j_1}\odot\ldots\odot T_D^{\otimes j_D}}{G}_{HS} = 0
$$
for all $j\in\{2,3,\ldots,k\}$. This equality holds for all $x_1,\ldots,x_D\in\real$, hence the coefficient of any monomial $x_1^{j_1}\ldots x_D^{j_D}$ must vanish. Since the operators
$$
\left\{ \Bin{j}{j_1 \ \ldots \ j_D}^{\frac{1}{2}} d^{\frac{j-k}{2}}\, \id^{\otimes (k-j)}\odot T_1^{\otimes j_1}\odot\ldots\odot T_D^{\otimes j_D} \mid 0\leq j \leq k ,\, j_1+\ldots+j_D = j \right\}
$$
constitute an orthonormal basis of $\Sym(k,\lh)$, it follows that
$$
G = a \id^{\otimes k} + \sum_{r=1}^D b_r \id^{\otimes (k-1)}\odot T_r = a \id^{\otimes k} + \id^{\otimes (k-1)}\odot T \,,
$$
where $a\in\real$ and $T\in\lh$ is a trace $0$ selfadjoint operator. By \eqref{eq:syst1},
$$
a = d^{-1} \tr{\A_i(x)} \,,
$$
and, by \eqref{eq:syst2},
$$
\tr{\Delta T} = \tr{\Delta\A_i(x)} \,.
$$
The last equation holds for all trace $0$ selfadjoint operators $\Delta$, hence $T = \A_i(x) - d^{-1}\tr{\A_i(x)} \,\id$. In conclusion,
$$
G = \id^{\otimes (k-1)}\odot \A_i(x) \,,
$$
which is \eqref{eq:joint-2}.
\end{proof}

Equation \eqref{eq:joint-2} should be compared with the usual compatibility, which requires that 
\begin{equation}\label{eq:joint-1}
\G^{[i]}(x) = \A_i(x) \qquad \forall i = 1,\ldots,n , \, x\in\Omega_i \, .
\end{equation}
There is one essential difference. 
While in the case of compatibility every joint observable satisfies \eqref{eq:joint-1}, in the case of $k$-compatibility not every joint observable satisfy \eqref{eq:joint-2} but there is always at least one which does.

\begin{corollary}\label{cor:kcomp}
The $\lh$-valued observables $\A_1,\ldots,\A_n$ are $k$-compatible if and only if the $\elle{\hh^{\otimes k}}$-valued observables $\widetilde{\A}_1,\ldots,\widetilde{\A}_n$ are compatible, where
\begin{equation*}
\widetilde{\A}_i(x) = \id^{\otimes (k-1)}\odot\A_i(x) \, .
\end{equation*}
\end{corollary}

\begin{example}
Let us consider two $2$-outcome observables $\A_1$ and $\A_2$ defined by positive operators $A_1$ and $A_2$, respectively. That is, $\Omega_1 = \Omega_2 = \{+1,-1\}$, and $\A_i(+1) = A_i$ for $i=1,2$. These are always 2-compatible, and a possible choice for their $2$-copy joint observable is given by \eqref{eq:kcompJM}. By Theorem \ref{prop:k-comp}, one can also find a \(\elle{\hi^{\otimes 2}} \)-valued symmetric joint observable \(\widetilde{\G}\) on \(\Omega_1\times\Omega_2\).
Indeed, if \(A_i^c=\id - A_i\), \(\widetilde{\G}\) is defined by 
\begin{align*}
\widetilde{\G}(+1,+1) & = \frac{1}{2}\left(A_1\otimes A_2 + A_2\otimes A_1\right)
 \\
\widetilde{\G}(-1,+1) & = \frac{1}{2}\left(A_1^c\otimes A_2 + A_2\otimes A_1^c \right)\\
\widetilde{\G}(+1,-1) & = \frac{1}{2}\left(A_1\otimes A_2^c +A_2^c\otimes A_1 \right)
\\
\widetilde{\G}(-1,-1)& = \frac{1}{2}\left(A_1^c\otimes A_2^c +A_2^c\otimes A_1^c \right)
 \end{align*}
\end{example}

\section{Three qubit observables}
\label{sec:qubitcase}

In this section we concentrate on the case of three observables.
Up to permutation of observables, there are six different compatibility stacks, depicted in Figure \ref{fig:possible}.
We will now show that all compatibility stacks in Figure \ref{fig:possible} have a realization in terms of qubit observables.

\subsection{2-copy joint observable from mixing }

Let $\X$, $\Y$ and $\Z$ be the three sharp spin-$1/2$ observables on \(\CC^2\), with outcome spaces $\{+1,-1\}$.
We further denote $\X_a(\pm 1) = \half ( \id \pm a \sigma_x)$ for $0\leq a \leq 1$, and similarly $\Y_b(\pm 1) = \half ( \id \pm b \sigma_y)$ and $\Z_c(\pm 1) = \half ( \id \pm c \sigma_z)$ for $0\leq b,c \leq 1$.
These are considered as noisy versions of the sharp observables $\X$, $\Y$ and $\Z$, with noise intensities \(1-a,1-b,\) and \(1-c\), respectively.

We recall the following results on joint measurability of noisy spin-$1/2$ observables \cite{Busch86}.

\begin{theorem}
\label{prop:qubit_compatibility}
The following facts hold.
\begin{itemize}[leftmargin=*,labelsep=3mm]
\item[(1)] $\X_a$ and $\Y_b$ are compatible if and only if $a^2+b^2 \leq 1$.
\item[(2)] $\X_a$, $\Y_b$ and $\Z_c$ are compatible if and only if $a^2+b^2+c^2 \leq 1$.
\end{itemize}
\end{theorem}

From these results we already find realizations of the cases (a)--(d) in Figure~\ref{fig:possible}. 
For instance, with the following choices of the parameters $a$, $b$ and $c$ we get suitable triplet of observables:
\begin{enumerate}[leftmargin=*,labelsep=3mm,label=(\alph*)]
\item $a=b=c=1/\sqrt{3}$;
\item $a=b=c=1/\sqrt{2}$;
\item $a=b=4/5$ and $c=3/5$;
\item $a=4/5$, $b=1$ and $c=3/5$.
\end{enumerate}

Let us then see how far we can get by mixing joint observables of two observables. 
The method is as follows. 
We choose randomly either $\X$, $\Y$ or $\Z$, measure the chosen observable, say $\X$, on the first system and then an optimal joint observable of the noisy versions of the remaining observables $\Y$ and $\Z$ on the second system. 
This gives the following sufficient condition for $2$-compatibility.

\begin{proposition}\label{prop:mixing}
(Sufficient condition for $2$-compatibility.)
$\X_a$, $\Y_b$ and $\Z_c$ are $2$-compatible if there are numbers $\lambda_1,\lambda_2,\lambda_3 \in [0,1]$ and $\alpha,\beta,\gamma \in [0,\pi/2]$ such that
\begin{equation*}
\lambda_1 + \lambda_2 + \lambda_3 = 1 
\end{equation*}
and
\begin{equation}\label{eq:random}
\begin{cases}
a \leq  \lambda_1 + \lambda_2 \cos\gamma + \lambda_3 \sin\alpha \\
b \leq \lambda_1 \sin\beta + \lambda_2 + \lambda_3 \cos\alpha \\
c \leq \lambda_1 \cos\beta + \lambda_2 \sin\gamma + \lambda_3  \\
\end{cases} \, .
\end{equation}
\end{proposition}

\begin{proof}
We choose randomly either $\X$, $\Y$ or $\Z$, measure the chosen observable, say $\X$, on the first system and then an optimal joint observables $\G^{2,3}$ of the noisy versions the remaining observables $\Y$ and $\Z$ on the second system. 
The total procedure leads to an observable
\begin{align*}
\G(x,y,z) = \lambda_1 \X(x) \otimes \G^{2,3}(y,z) + \lambda_2 \Y(y) \otimes \G^{1,3}(x,z) + \lambda_3 \Z(z) \otimes \G^{1,2}(x,y) \, ,
\end{align*}
where $\lambda_1,\lambda_2, \lambda_3$ represent the probabilities for the choice of the first measurement $\X,\Y,\Z$, respectively.
The marginals of $\G^{1,2}$, $\G^{1,3}$ and $\G^{2,3}$ are noisy observables and their noise intensities are limited by Theorem \ref{prop:qubit_compatibility}.
Calculating the marginals of $\G$ gives $\X_a$, $\Y_b$ and $\Z_c$ with $a,b,c$ limited by \eqref{eq:random}.
\end{proof}

Using Proposition \ref{prop:mixing} we can cook up a realization of Figure~\ref{fig:possible}.(e).
The $2$-compatibility of $\X_a$, $\Y_b$ and $\Z_c$ can be achieved by setting
 $\alpha=\beta=\gamma=\pi/4$ and $\lambda_1=\lambda_2=\lambda_3=1/3$.
This choice provides a possible setting (in Figure~\ref{fig:possible})
\begin{enumerate}[leftmargin=*,labelsep=3mm,label=(\alph*),start=5]
\item $a=b=c=(1+\sqrt{2})/3$.
\end{enumerate}

\subsection{Optimal 2-copy joint observable}\label{app:proof_teo:2comp}

To find a realization of the compatibility stack depicted in Figure~\ref{fig:possible}.(f), we need to show that $\X_a$, $\Y_b$, $\Z_c$ are not $2$-compatible for some values of noise intenstities $a,b,c$. 
The fact that these kind of parameters exist follows from the next theorem.

\begin{theorem}\label{teo:2comp}
$\X_a$, $\Y_a$ and $\Z_a$ are 2-compatible if and only if $0\leq a\leq \sqrt{3}/2$.
\end{theorem}

By Theorem \ref{prop:k-comp}, the observables $\X_a$, $\Y_a$ and $\Z_a$ are $2$-compatible if and only if the observables $\widetilde{\X}_a$, $\widetilde{\Y}_a$ and $\widetilde{\Z}_a$ have a symmetric joint observable $\widetilde{\G}$, where
\begin{align}
\widetilde{\X}_a(\pm 1) & = \id\odot\X_a(\pm 1) = \half (\id \otimes \X_a(\pm 1) + \X_a(\pm 1) \otimes \id ) \notag\\
&= \frac{1}{4} ( 2 \, \id \otimes \id \pm a (\id \otimes \sigma_x + \sigma_x\otimes \id ) ) \label{eq:tildeX}
\end{align}
and similarly for $\widetilde{\Y}_a$ and $\widetilde{\Z}_a$. We will now show that $\widetilde{\G}$ can be chosen to be covariant with respect to the transitive action of a suitable group on the joint outcome space $\Omega = \{+1,-1\}^3$ of the three observables $\widetilde{\X}_a$, $\widetilde{\Y}_a$ and $\widetilde{\Z}_a$. Covariance will then drastically decrease the freedom in the choice of $\widetilde{\G}$, actually reducing it to only fixing two parameters.
To exploit covariance, we start from the following simple fact.

\begin{proposition}\label{prop:covariantization}
Suppose $G$ is a finite group, $\Omega$ is a $G$-space and $U$ is a unitary representation of $G$ in the Hilbert space $\kk$. Let $\ff$ be a collection of subsets of $\Omega$ such that
$$
g.X = \{g.x\mid x\in X\}\in\ff \quad \text{for all $g\in G$ and $X\in\ff$} \,.
$$
Then, for any observable $\G:\Omega\to\lk$ satisfying the relation
$$
\G(g.X) = U(g)\G(X)U(g)^* \qquad \forall X\in\ff\,,\,g\in G \,,
$$
the observable $\widehat{\G}:\Omega\to\lk$ given by
\begin{equation}\label{eq:covariantization}
\widehat{\G}(x) = \frac{1}{\# G} \sum_{g\in G} U(g)^*\G(g.x)U(g) \qquad \forall x\in\Omega
\end{equation}
is such that
\begin{enumerate}[leftmargin=*,labelsep=3mm,label=(\roman*)]
\item $\widehat{\G}(g.x) = U(g)\widehat{\G}(x)U(g)^*$ for all $x\in\Omega$ and $g\in G$;
\item $\widehat{\G}(X) = \G(X)$ for all $X\in\ff$.
\end{enumerate}
\end{proposition}
\begin{proof}
Direct verification.
\end{proof}

According to \eqref{eq:covariantization}, we call the observable $\widehat{\G}$ the {\em $U$-covariant version} of $\G$.

The choice of the covariance group $G$ and its action on the outcomes $\Omega$ for a joint observable of $\widetilde{\X}_a$, $\widetilde{\Y}_a$ and $\widetilde{\Z}_a$ is prescribed by the covariance properties of $\X$, $\Y$ and $\Z$. Namely, the set of effects $\{\X(\pm 1),\Y(\pm 1),\Z(\pm 1)\}$ is invariant for the rotations in the octahedron subgroup $O\subset SO(3)$. Moreover, $O$ acts transitively on this set. We therefore expect that the proper covariance group for our problem is $G=O$. We now better explain this statement.

The {\em octahedron group} $O$ is the order $24$ group of the $90^\circ$ rotations around the three coordinate axes $\vec{i},\vec{j},\vec{k}$, together with the $120^\circ$ rotations around the axes $(\pm\vec{i}\pm\vec{j}\pm\vec{k})/\sqrt{3}$ and the $180^\circ$ rotations around $(\pm\vec{i}\pm\vec{j})/\sqrt{2}$, $(\pm\vec{j}\pm\vec{k})/\sqrt{2}$ and $(\pm\vec{i}\pm\vec{k})/\sqrt{2}$. It preserves the set $\Omega = \{(x,y,z)\in\RR^3 \mid x,y,z\in\{+1,-1\}\}$ and acts transitively on it. Moreover, the stabilizer subgroup of any $\vec{u}\in\Omega$ is just the subgroup $O_{\vec{u}}$ of the three $120^\circ$ rotations around $\vec{u}/\sqrt{3}$.

The octahedron group also acts on the spin-$1/2$ Hilbert space $\hh=\CC^2$ by restriction of the usual two-valued $SU(2)$-representation of $SO(3)$. This gives an ordinary representation $U(g) = \tilde{g}\otimes\tilde{g}$ of $O$ on the $2$-copy Hilbert space $\hh^{\otimes 2} = \CC^2\otimes\CC^2$, where $\tilde{g}$ is any of the two elements of $SU(2)$ lying above $g\in O$.

Finally, let $\pi_i : \Omega\to\{+1,-1\}$ be the projection onto the $i$-th component ($i=1,2,3$), and define the collection of subsets
\begin{align*}
\ff = & \{\pi_1^{-1}(x) \,,\, \pi_2^{-1}(y) \,,\, \pi_3^{-1}(z) \mid x,y,z\in\{+1,-1\}\} \,.
\end{align*}
Clearly, the collection $\ff$ is $O$-invariant. Moreover, if $\widetilde{\G}:\Omega\to\lh$ is any symmetric joint observable of $\widetilde{\X}_a$, $\widetilde{\Y}_a$ and $\widetilde{\Z}_a$, then
$$
\widetilde{\G}(\pi_1^{-1}(x)) = \widetilde{\X}_a(x) \qquad \widetilde{\G}(\pi_2^{-1}(y)) = \widetilde{\Y}_a(y) \qquad \widetilde{\G}(\pi_3^{-1}(z)) = \widetilde{\Z}_a(z) \,.
$$
The covariance properties of the observables $\X$, $\Y$ and $\Z$ then imply that $\widetilde{\G}(g.X) = U(g)\widetilde{\G}(X)U(g)^*$ for all $X\in\ff$ and $g\in O$. Hence, by Proposition \ref{prop:covariantization} the $U$-covariant version $(\widetilde{\G})^\wedge$ of $\widetilde{\G}$ defined in \eqref{eq:covariantization} yields the same margins $\widetilde{\X}_a$, $\widetilde{\Y}_a$ and $\widetilde{\Z}_a$. Since the representations $U$ of $O$ and $\sigma$ of $S_2$ commute, the joint observable $(\widetilde{\G})^\wedge$ is both $U$-covariant and symmetric.

In summary, in order to find the maximal value of $a$ for which the observables $\X_a$, $\Y_a$ and $\Z_a$ are $2$-compatible, we are led to classify the family of symmetric $U$-covariant observables on $\Omega$. This is done in the next proposition.

\begin{proposition}\label{prop:Ucovar}
A map $\G:\Omega\to\elle{\hh^{\otimes 2}}$ is a symmetric and $U$-covariant observable if and only if there exist real numbers $\alpha$ and $\beta$ with $\alpha\geq 0$, $\beta\geq 0$ and $\alpha+\beta\leq 3/8$ such that
\begin{align}\label{eq:covar}
\begin{aligned}
\G(\vec{u}) = & \frac{4(\alpha+\beta)-1}{16} \left[\vec{u}\cdot\vec{\sigma}\otimes\vec{u}\cdot\vec{\sigma} - (\sigma_x\otimes\sigma_x + \sigma_y\otimes\sigma_y + \sigma_z\otimes\sigma_z)\right] \\
& + \frac{\alpha-\beta}{4\sqrt{3}} \left(\vec{u}\cdot\vec{\sigma}\otimes\id + \id\otimes\vec{u}\cdot\vec{\sigma}\right) + \frac{1}{8} \id\otimes\id
\end{aligned}
\end{align}
for all $\vec{u}\in\Omega$.
\end{proposition}
\begin{proof}
We will proceed in several steps.

(I) Since the action of $O$ on $\Omega$ is transitive, a $U$-covariant observable $\G$ is completely determined by its value at $\vec{u}_0 = (+1,+1,+1)$ by the relation
\begin{align}
\G(g.\vec{u}_0) = U(g) \G(\vec{u}_0) U(g)^* \qquad \forall g\in O \,. \label{eq:covar2}
\end{align}
This equation implies that $\G(\vec{u}_0)$ must commute with the representation $U$ restricted to the stabilizer $O_{\vec{u}_0}$ of $\vec{u}_0$. This happens if and only if $\G(\vec{u}_0)$ is in the commutant $(\e^{i(\pi/3)\vec{n}\cdot\vec{\sigma}}\otimes\e^{i(\pi/3)\vec{n}\cdot\vec{\sigma}})'$ of the operator $\e^{i(\pi/3)\vec{n}\cdot\vec{\sigma}}\otimes\e^{i(\pi/3)\vec{n}\cdot\vec{\sigma}}$, where $\vec{n} = \vec{u}_0/\sqrt{3} = (\vec{i}+\vec{j}+\vec{k})/\sqrt{3}$. The eigenvalues of $\e^{i(\pi/3)\vec{n}\cdot\vec{\sigma}}\otimes\e^{i(\pi/3)\vec{n}\cdot\vec{\sigma}}$ are $\e^{i(2\pi/3)}$ and $\e^{-i(2\pi/3)}$ with multiplicity one, and $1$ with multiplicity two. Hence, $\dim (\e^{i(\pi/3)\vec{n}\cdot\vec{\sigma}}\otimes\e^{i(\pi/3)\vec{n}\cdot\vec{\sigma}})' = 6$. A linear basis of $(\e^{i(\pi/3)\vec{n}\cdot\vec{\sigma}}\otimes\e^{i(\pi/3)\vec{n}\cdot\vec{\sigma}})'$ is made up of the selfadjoint opertors
\begin{gather*}
M_0 = \id\otimes\id \qquad M_1 = \vec{n}\cdot\vec{\sigma}\otimes\vec{n}\cdot\vec{\sigma} \\
M_2 = \frac{1}{3}(\sigma_x\otimes\sigma_x + \sigma_y\otimes\sigma_y + \sigma_z\otimes\sigma_z) \\
M_3 = \vec{n}\cdot\vec{\sigma}\otimes\id + \id\otimes\vec{n}\cdot\vec{\sigma} \qquad
M_4 = \vec{n}\cdot\vec{\sigma}\otimes\id - \id\otimes\vec{n}\cdot\vec{\sigma} \\
M_5 = \sigma_x\otimes\sigma_y - \sigma_y\otimes\sigma_x + \sigma_y\otimes\sigma_z - \sigma_z\otimes\sigma_y + \sigma_z\otimes\sigma_x - \sigma_x\otimes\sigma_z \,.
\end{gather*}
Among them, $M_0$, $M_1$, $M_2$ and $M_3$ are symmetric, and $M_4$ and $M_5$ are antisymmetric. Thus, $\G$ is symmetric only if $\G(\vec{u}_0)$ is a real linear combination of $M_0$, $M_1$, $M_2$, $M_3$. This is also a sufficient condition for the symmetry of $\G$ by \eqref{eq:covar2} and the symmetry of the $U(g)$'s.

(II) It is easy to check that the operators $M_0$, $M_1$, $M_2$ and $M_3$ all commute among themselves. Moreover,
\begin{gather*}
M_0^2 = M_1^2 = M_0\,, \qquad M_2^2 = \frac{1}{3} (M_0 - 2 M_2)\,, \qquad M_3^2 = 2(M_0+M_1)\,,\\
M_1 M_2 = \frac{1}{3}(M_0+M_1-3M_2)\,, \qquad M_1 M_3 = M_3\,, \qquad M_2 M_3 = \frac{1}{3} \, M_3 \,,\\
M_0 M_i = M_i \quad \forall i\,.
\end{gather*}
Thus, the four selfadjoint operators
\begin{align*}
P_+ & = \frac{1}{4} (M_0 + M_1 + M_3) & \qquad
P_- & = \frac{1}{4} (M_0 + M_1 - M_3) \\
Q_+ & = \frac{1}{4} (M_0 - 3M_2)  & \qquad
Q_- & = \frac{1}{4} (M_0 - 2M_1 + 3M_2)
\end{align*}
are mutually commuting orthogonal projections summing up to the identity of $\hh^{\otimes 2}$. It follows that $P_+$, $P_-$, $Q_+$, $Q_-$ are rank-$1$ mutually orthogonal projections. Since they span the same linear space as $\{M_0,M_1,M_2,M_3\}$, we can rewrite
\begin{equation}\label{eq:M111}
\G(\vec{u}_0) = \alpha P_+ + \beta P_- + \gamma Q_+ + \delta Q_-
\end{equation}
where
\begin{equation}\label{eq:pos1}
\alpha\geq 0\,, \qquad \beta\geq 0\,, \qquad \gamma\geq 0\,, \qquad \delta\geq 0
\end{equation}
by the positivity condition $\G(\vec{u}_0)\geq 0$.

(III) By taking the trace of the normalization condition
\begin{equation}\label{eq:norm}
\sum_{gO_{\vec{u}_0}\in O/O_{\vec{u}_0}} U(g)\G(\vec{u}_0)U(g)^* = \id\otimes\id
\end{equation}
and observing that $\# O/O_{\vec{u}_0} = \#\Omega = 8$, we obtain
\begin{equation}\label{eq:cond1}
\tr{\G(\vec{u}_0)} = \frac{1}{2} \,.
\end{equation}
Moreover, the operators $M_0$ and $M_2$ commute with the representation $U$, hence so does the rank-$1$ projection $Q_+$. Multiplying both the sides of \eqref{eq:norm} by $Q_+$ and taking again the trace, we then get
\begin{equation}\label{eq:cond2}
\tr{Q_+\G(\vec{u}_0)} = \frac{1}{8} \,.
\end{equation}
Inserting \eqref{eq:M111} into \eqref{eq:cond1} and \eqref{eq:cond2} yields the conditions
$$
\gamma = \frac{1}{8} \qquad \qquad \delta = \frac{3}{8} - \alpha - \beta \,.
$$
The positivity requirement \eqref{eq:pos1} thus translates into
\begin{equation*}
\alpha\geq 0\,, \qquad \beta\geq 0\,, \qquad \alpha+\beta\leq\frac{3}{8} \,,
\end{equation*}
and \eqref{eq:M111} is rewritten as
\begin{eqnarray*}
\G(\vec{u}_0) & = & \alpha (P_+ - Q_-) + \beta (P_- - Q_-) + \frac{1}{8}\,Q_+ + \frac{3}{8}\,Q_- \\
& = & \frac{3}{4}\left(\alpha+\beta-\frac{1}{4}\right) (M_1-M_2) + \frac{1}{4}\left(\alpha-\beta\right) M_3 + \frac{1}{8} M_0 \\
& = & \frac{4(\alpha+\beta)-1}{16} \left[\vec{u}_0\cdot\vec{\sigma}\otimes\vec{u}_0\cdot\vec{\sigma} - (\sigma_x\otimes\sigma_x + \sigma_y\otimes\sigma_y + \sigma_z\otimes\sigma_z)\right] \\
&& + \frac{\alpha-\beta}{4\sqrt{3}} \left(\vec{u}_0\cdot\vec{\sigma}\otimes\id + \id\otimes\vec{u}_0\cdot\vec{\sigma}\right) + \frac{1}{8} \id\otimes\id \,.
\end{eqnarray*}
We have already seen that $3M_2 = \sigma_x\otimes\sigma_x + \sigma_y\otimes\sigma_y + \sigma_z\otimes\sigma_z$ commutes with $U(g)$ for all $g\in O$. Therefore, formula \eqref{eq:covar} follows from the previous equation by the relation \eqref{eq:covar2}.
\end{proof}

\begin{remark}\label{rem:covar}
The choice of the covariance group $G = O$ and its natural action on the joint outcome space $\Omega$ is the minimal possible in order to construct a transitive action of $G$ on $\Omega$ preserving the set of effects $\{\widetilde{\X}_a(\pm 1), \widetilde{\Y}_a(\pm 1), \widetilde{\Z}_a(\pm 1)\}$. Transitivity is needed in order to label all the covariant joint observables by means of the single operator $\G(\vec{u}_0)$ as in \eqref{eq:covar2}, and thus reduce the many free parameters of the problem to the only choice of such an operator.
\end{remark}

Now, we need to take the three margins of the most general $U$-covariant observable found in Proposition \ref{prop:Ucovar} and compare it with the observables $\widetilde{\X}_a$, $\widetilde{\Y}_a$ and $\widetilde{\Z}_a$. By the covariance property, it is sufficient to consider only the first margin $\G^{[1]}$. We have
\begin{align*}
& \sum_{y,z\in\{+1,-1\}} (x\sigma_x+y\sigma_y+z\sigma_z)\otimes(x\sigma_x+y\sigma_y+z\sigma_z) = 4(\sigma_x\otimes\sigma_x + \sigma_y\otimes\sigma_y + \sigma_z\otimes\sigma_z) \\
& \sum_{y,z\in\{+1,-1\}} [(x\sigma_x+y\sigma_y+z\sigma_z)\otimes\id + \id\otimes(x\sigma_x+y\sigma_y+z\sigma_z)] = 4x(\sigma_x\otimes\id+\id\otimes\sigma_x)
\end{align*}
and hence
$$
\G^{[1]}(x) = \sum_{y,z\in\{+1,-1\}} \G(x,y,z) = \frac{(\alpha-\beta)x}{\sqrt{3}} (\sigma_x\otimes\id+\id\otimes\sigma_x) + \frac{1}{2}\,\id\otimes\id \,.
$$
Comparing this formula with \eqref{eq:tildeX} yields
$$
\alpha-\beta = \frac{\sqrt{3}}{4} \, a \,.
$$
By the positivity conditions $\alpha\geq 0$, $\beta\geq 0$ and $\alpha+\beta\leq 3/8$, we thus see that the maximal value of $a$ is $a=\sqrt{3}/2$. This completes the proof of Theorem \ref{teo:2comp}.

As a result, the case of Figure~\ref{fig:possible}.(f) can now be achieved for example by setting
\begin{enumerate}[leftmargin=*,labelsep=3mm,label=(\alph*),start=6]
\item $a=b=c=1$,
\end{enumerate}
though any choice larger than $\sqrt{3}/2$ suffices.

By combining the result of Theorem \ref{teo:2comp} with that of Theorem \ref{prop:qubit_compatibility} in the case $a=b=c$, we get a complete characterization of the index of incompatibility for three equally noisy orthogonal qubit observables $\X_a$, $\Y_a$, and $\Z_a$. In Figure~\ref{fig:3_qubit_index}, the index of incompatibity $\ind{\{\X_a,\Y_a,\Z_a\}}$ is plotted as a function of the noise parameter $a$.

\begin{figure}
\centering
 \includegraphics[width=8cm]{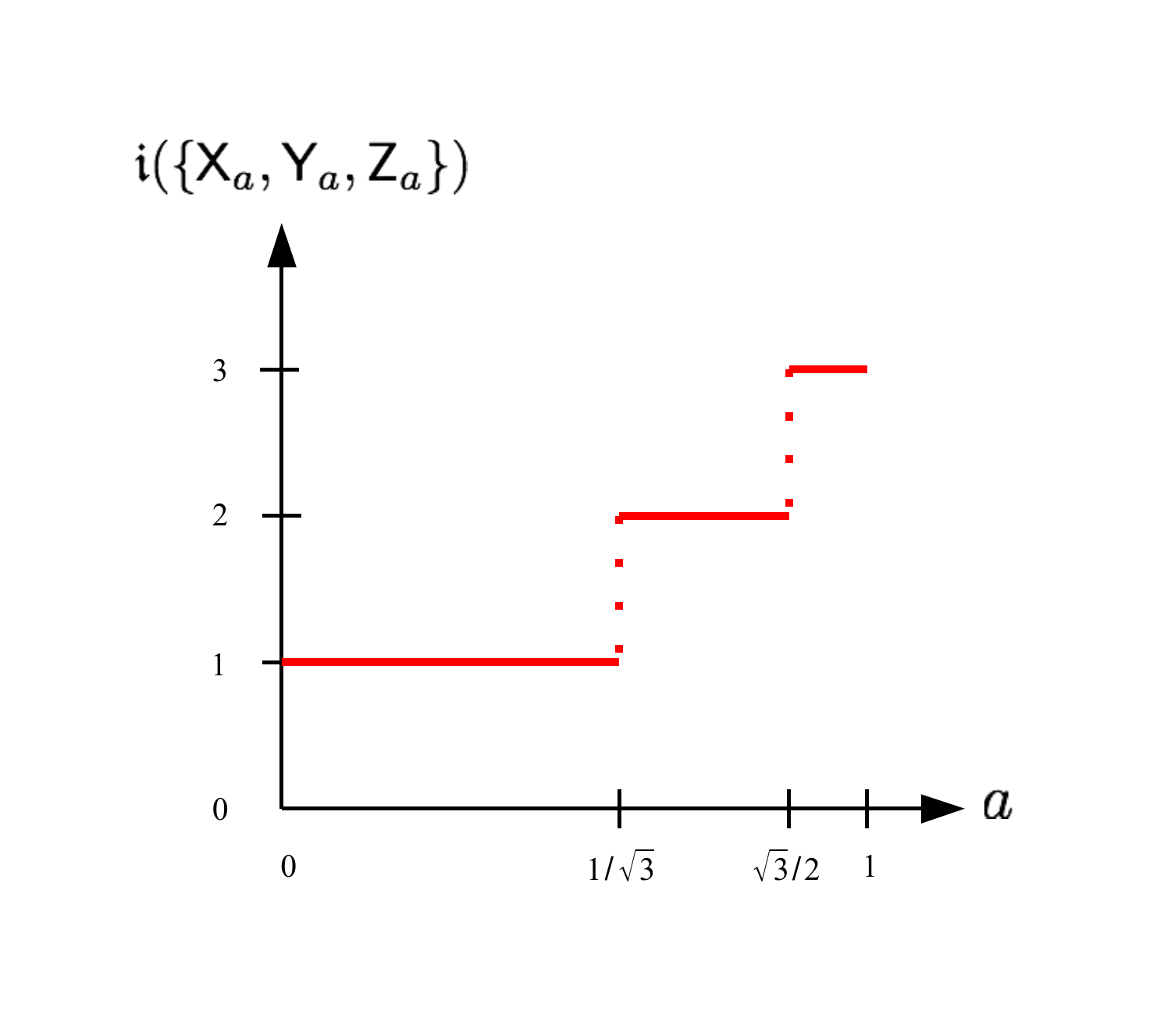}
 \caption{\label{fig:3_qubit_index}The index of incompatibility $\ind{\{\X_a,\Y_a,\Z_a\}}$ as a function of the noise parameter $a$ for three noisy orthogonal qubit observables.}
\end{figure}

\section{Concluding remarks}\label{sec:conclusion}


The incompatibility of quantum observables can be evaluated and measured in various ways. In this paper we introduce a measure of incompatibility based on the number of system copies needed to be able to measure the given observables simultaneously. We call this number the index of  incompatibility.  It quantifies the incompatibility of a set of observables as a whole, but leaves out the finer details regarding the various compatibility relations between the observables.

In \cite{KuHeFr14} it was shown that every conceivable joint measurability combination of a set of observables is realizable. Such combinations are representable by joint measurability hypergraphs where the vertices are observables and edges mark the compatibility relations. 
By translating this approach to our multi-copy setting, we have  analogously defined the notion of compatibility stack that represents the potential multi-copy compatibility relations present in a set of observables. Namely, whereas in \cite{KuHeFr14} the hypegraph was binary --- the presence of graph edges represented compatibility between the observables of corresponding subset --- here we have such a graph for each possible compatibility index. 
We demonstrate that every compatibility stack with three vertices has a realization in terms of quantum observables. However, it remains an open question if all compatibility stacks have such a realization.

\vspace{6pt} 


\paragraph{\textbf{Acknowledgments}} {A.T. acknowledges financial support from the Italian Ministry of Education, University and Research (FIRB project RBFR10COAQ).}

\end{document}